\documentclass[10pt,reqno]{amsart}
     \makeatletter
     \def\section{\@startsection{section}{1}%
     \z@{.7\linespacing\@plus\linespacing}{.5\linespacing}%
     {\bfseries
     \centering
     }}
     \def\@secnumfont{\bfseries}
     \makeatother
\newtheorem{theorem}{Theorem}[section]

\newtheorem{proposition}[theorem]{Proposition}

\theoremstyle{definition}
\newtheorem{definition}[theorem]{Definition}

\theoremstyle{remark}
\newtheorem{remark}[theorem]{Remark}
\numberwithin{equation}{section}
\usepackage{mathtools}
\usepackage{float}
\begin{document}

\title[Uncertain Volatility]{Modelling Uncertain Volatility Using Quantum Stochastic Calculus: Unitary vs Non-Unitary Time Evolution}

\author{Will Hicks}
\address{Will Hicks: Centre for Quantum Social and Cognitive Science, Memorial University of Newfoundland}
\email{williamh@mun.ca}



\begin{abstract}
In this article we look at stochastic processes with uncertain parameters, and consider different ways in which information is obtained when carrying out observations. For example we focus on the case of a the random evolution of a traded financial asset price with uncertain volatility. The quantum approach presented, allows us to encode different volatility levels in a state acting on a Hilbert space. We consider different means of defining projective measurements in order to track the evolution of a traded market price, and discuss the results of different Monte-Carlo simulations.
\end{abstract}

\maketitle


\section{Introduction}
In \cite{AB}, and \cite{WH}, the authors consider the stochastic evolution of a price variable in the tensor product of a Hilbert space with a Boson Fock space (containing the random noise element of the stochastic evolution):
\begin{align}\label{Hilbert_1}
\mathcal{H}&=\mathcal{H}_{mkt}\otimes\Gamma(L^2(\mathbb{R}^+;\mathbb{C}))
\end{align}
The observables considered are of the form:
\begin{align*}
X\otimes\mathbb{I}
\end{align*}
where $X$ is a self-adjoint operator acting on the Hilbert space: $\mathcal{H}_{mkt}$. Time evolution is via unitary operators, and we work in the Heisenberg interpretation, so that at time $t$, the observables take the form:
\begin{align}\label{UXU}
j_t(X)&=U_t^*(X\otimes\mathbb{I})U_t
\end{align}
where, as shown in \cite{HP}, the unitary operators $U_t$ satisfy:
\begin{align}\label{unitary}
dU_t&=-\Big(\big(iH+\frac{L^*L}{2}\big)dt+L^*dA_t-LdA^{\dagger}_t\Big)U_t
\end{align}
Here, $dA_t$ and $dA^{\dagger}_t$ introduce random noise in the space $\Gamma(L^2(\mathbb{R}^+;\mathbb{C}))$, so are strictly given by: $1\otimes dA_t$, and $1\otimes dA^{\dagger}_t$. $H$ is the system Hamiltonian, and $L$ scales the degree of variance introduced by the stochastic process into the noise space: $\Gamma\big(L^2(\mathbb{R}^+;\mathbb{C})\big)$ (we abbreviate as $\Gamma$). For example in \cite{WH} the author uses:
\begin{align}\label{L_simple}
\mathcal{H}_{mkt}&=L^2(\mathbb{R})\nonumber\\
L&=-i\sigma\frac{\partial}{\partial x}\otimes\mathbb{I}
\end{align}
This leads to a Gaussian process for $j_t(X)$ with total variance: $\sigma^2t$.

In some cases, for example when looking at a physical diffusion process, one may very well be able to measure the diffusivity, which is therefore known in advance. In financial markets, this is not really the case. When one observes traded market prices evolving, one does not know and cannot really measure, what the value of $\sigma$ is, versus what the precise value for the random noise element is.

When pricing financial derivatives, to overcome this, one often makes an assumption regarding the stochastic processes involved, before calibrating parameters, such as the variance, to some external source. For example, the well known case of the Black-Scholes implied volatility derived from the prices of exchange traded options.

Once one has chosen the form of the stochastic process, and fixed the parameters by some calibration process, knowledge of the change in the market price fixes the value of the Gaussian noise (and vice versa).

The Hilbert space framework of quantum probability provides a method by which one can simulate price changes without fixing the value of uncertain parameters. In this article, we consider the inclusion of an additional Hilbert space that reflects the state of market volatility:
\begin{align}\label{Hilbert_2}
\mathcal{H}&=\mathcal{H}_{mkt}\otimes\mathcal{H}_{\sigma}\otimes\Gamma(L^2(\mathbb{R}^+;\mathbb{C}))
\end{align}
We consider how to define a projective measurement on this Hilbert space, and how the nature of the information that we collect from the volatility space determines the dynamics of the process as we model the time evolution into the future.
\section{Hamilton Jacobi Bellman Approach:}\label{HJB_app}
In \cite{Av}, and \cite{Lyons}, the authors address the problem of pricing options on a risky asset that satisfies a stochastic differential equation with uncertain volatility:
\begin{align*}
dX_t&=\sigma_tX_tdW_t\\
\sigma_t&\in [\sigma_{min},\sigma_{max}]
\end{align*}
They look for a solution that maximises the valuation of an option, with a view to presenting a sell side price. In other words, a price that presents a high likelihood of being able to yield a profit. This approach leads to a Hamilton Jacobi Bellman equation:
\begin{align*}
\partial_tu(t,x)&+\frac{1}{2}x^2\Sigma\big(\partial_x^2u(t,x)\big)^2\partial_x^2u(t,x)=0\\
\Sigma(\Gamma)&=\sigma_{min}1_{(\Gamma<0)}+\sigma_{max}1_{(\Gamma>0)}
\end{align*}
In this article, we consider a similar problem. That is where the volatility is defined by a quantum state acting on a Hilbert space. This enables us to identify a risk neutral price. We go on to show how the model can be calibrated to market prices in practice, whilst allowing the dynamics of the model to be defined by the quantum state acting on the space $\mathcal{H}_{\sigma}$ and a system Hamiltonian.

However, we note here the similarity of the underlying assumption. That is that the risky underlying follows a Brownian motion, but where the volatility is not known in advance.
\section{Uncertain Volatility: Quantum Stochastic Approach}
\subsection{Finite Dimensional Volatility Space:}
In this section, we consider using the finite dimensional Hilbert space: $\mathbb{C}^K$ for $\mathcal{H}_{\sigma}$, and $L^2(\mathbb{R})$ for $\mathcal{H}_{mkt}$. We write the eigenvectors for $\mathcal{H}_{\sigma}$ as: $\{|s_i\rangle\text{, }i=1\dots k\}$. In this instance we can change the operator \eqref{L_simple} that defines the stochastic process to be:
\begin{align}\label{L_no_Ham}
L&=-i\frac{\partial}{\partial x}\otimes\sum_{k=1}^K\sigma_k|s_k\rangle\langle s_k|\otimes\mathbb{I}
\end{align}
where $\sigma_k$ is the eigenvalue associated with eigenvector $|s_k\rangle$. In other words, rather than scaling the variance added to the noise space: $\Gamma$ by a fixed volatility parameter: $\sigma$, as in equation \eqref{L_simple}, we allow the action of the state on $\mathcal{H}_{\sigma}$ to define the level of the volatility. This of course allows for the possibility that the volatility is not known in advance. Under the new Hilbert space structure \eqref{Hilbert_2}, equation \eqref{UXU} becomes:
\begin{align}\label{UXU2}
j_t(X)&=U_t^*(X\otimes\mathbb{I}\otimes\mathbb{I})U_t\\
X\psi(x)&=(x\psi)(x)\text{, }\psi(x)\in L^2(\mathbb{R})\nonumber
\end{align}
Note that $X$ acts as a multiplication operator on the market Hilbert space: $\mathcal{H}_{mkt}$. Under the unitary time evolution defined by equation \eqref{unitary}, the operator \eqref{UXU} satisfies the following (see for example \cite{AB}, \cite{HP}):
\begin{align}\label{dX}
dj_t(X)&=j_t(\alpha)dA_t+j_t(\alpha^{\dagger})dA^{\dagger}_t\\
\alpha&=[L^*,X]\nonumber\\
dj_t(X^2)&=j_t(\alpha\alpha^{\dagger})dt\nonumber
\end{align}
In the simple case \eqref{L_simple} we get:
\begin{align}\label{Com_Simple}
[L^*,X]&=-i\sigma\\
dj_t(X^2)&=\sigma^2dt\nonumber
\end{align}
For the case \eqref{L_no_Ham}/\eqref{UXU2}, we have:
\begin{align}\label{Com_noHam}
[L^*,X]&=-i\otimes\sum_{k=1}^K\sigma_k|s_k\rangle\langle s_k|\otimes\mathbb{I}\\
j_t(X^2)&=\mathbb{I}\otimes\sum_{k=1}^K\sigma_k^2|s_k\rangle\langle s_k|\otimes dt
\end{align}
Note that the variance for the simple process, defined by equation \eqref{L_simple}, is not dependent on the quantum state acting on $\mathcal{H}_{mkt}\otimes\Gamma$. However, for the model defined by the Hilbert space \eqref{Hilbert_2}, and equation \eqref{L_no_Ham}, we require knowledge of the quantum state to configure the different possibilities for the resulting variance. For simplicity, we assume the overall state is a product state:
\begin{align}\label{prod_state}
\rho&=\rho_{mkt}\otimes\rho_{\sigma}\otimes|\psi_{\Gamma}(0)\rangle\langle\psi_{\Gamma}(0)|
\end{align}
where we assume the Fock space starts in the vacuum state: $|\psi_{\Gamma}(0)\rangle\langle\psi_{\Gamma}(0)|$. We can now consider different options for $\rho_{\sigma}$. First, we might have:
\begin{align*}
\rho_{\sigma}&=\sigma |s_n\rangle\langle s_n|
\end{align*}
in which case, the model reduces to that defined by equation \eqref{Hilbert_1} and \eqref{L_simple}:
\begin{align*}
dj_t(X^2)&=\sigma^2dt
\end{align*}
Alternatively, we could have a maximum entropy state:
\begin{align}\label{rho_sig}
\rho_{\sigma}&=\frac{1}{K}\sum_{k=1}^K|s_k\rangle\langle s_k|
\end{align}
In fact, to understand the behaviour of the model under the state \eqref{rho_sig}, we must define how much information about the quantum state we retain as we model further into the future. This is discussed further in section: \eqref{Comment_U_NU}.
\subsection{Projective Measurements:}
We assume that we have a stochastic process: equation \eqref{UXU2}, where the volatility is given by equation \eqref{L_no_Ham}. That is, we have a Brownian diffusion with uncertain volatility. In this article we consider a Monte-Carlo simulation of this process, whereby the price is measured at finite intervals. There are two processes by which the state of our system (observables plus state) can change. The first is via the unitary time evolution that introduces noise into the Boson Fock space. The second is via a measurement process, which we define in proposition \eqref{gen_meas_prop}. Before doing so, we define the hybrid interpretation that we will be using (in place of the Schr{\"o}dinger picture or Heisenberg picture).
\begin{proposition}{Hybrid Interpretation:}\label{hybrid_prop}
Assume that the initial state is given by (per equation \eqref{prod_state}):
\begin{align*}
\rho(0)&=\rho_{mkt}\otimes\rho_{\sigma}\otimes|\psi(0)\rangle\langle\psi(0)|
\end{align*}
Assume that the Hilbert space is given by equation \eqref{Hilbert_2}, and the unitary time evolution by equations \eqref{unitary} and \eqref{L_no_Ham}. Furthermore, assume that the measurements are represented by a projection operator, that we label: $\mathcal{P}_x$, for the measured value $x$. Finally we assume that we take measurements at time intervals: $\tau_n$, where we have:
\begin{align*}
T_n&=\sum_{i\leq n}\tau_i
\end{align*}
Then the probability of measuring the value $j_{T_n}(X)=x$, conditional on the measured values: $y_i$ at $T_i: i=\{1,\dots n-1\}$, is given by:
\begin{align}\label{hybrid}
p(x|y_i,i=1,\dots &,n-1)=Tr[\rho(T_{n-1})j_{T_n}(\mathcal{P}_x)]\\
\rho(T_{n-1})&=\frac{j_{T_{n-1}}(\mathcal{P}_{y_{n-1}})\dots j_{T_1}(\mathcal{P}_{y_1})\rho(0)j_{T_1}(\mathcal{P}_{y_1})\dots j_{T_{n-1}}(\mathcal{P}_{y_{n-1}})}{Tr[j_{T_{n-1}}(\mathcal{P}_{y_{n-1}})\dots j_{T_1}(\mathcal{P}_{y_1})\rho(0)j_{T_1}(\mathcal{P}_{y_1})\dots j_{T_{n-1}}(\mathcal{P}_{y_{n-1}})]}\nonumber
\end{align}
\end{proposition}
\begin{proof}
For the first time-step: $\tau_1$, we have in the Schr{\"o}dinger interpretation, the probability of obtaining the value $x$ at $T_1$ is given by:
\begin{align*}
p(x)&=Tr[U_{\tau_1}\rho(0)U_{\tau_1}^*\mathcal{P}_x]\\
&=Tr[\rho(0)U_{\tau_1}^*\mathcal{P}_xU_{\tau_1}]\\
&=Tr[\rho(0)j_{\tau_1}(\mathcal{P}_x)]\\
&=Tr[\rho(0)j_{T_1}(\mathcal{P}_x)]
\end{align*}
Therefore, the proposition holds for the first time-step. Now we assume that the proposition holds up to $T_{n-2}$. That is we assume the state at $T_{n-2}$ is given by the formula in equation \eqref{hybrid}, and we assume the operators at $T_{n-2}$ are given by: $j_{T_{n-2}}(\mathcal{P}_y)$.

For the time period $T_{n-2}$ to $T_{n-1}$ we work out the required probabilities using time evolution in the Schr{\"o}dinger interpretation, and show that this is equivalent to \eqref{hybrid}. The proposition then stands by induction.

In the Schr{\"o}dinger interpretation, the time evolution in the state from $T_{n-2}$ to $T_{n-1}$ is given by (we write $\rho^S$ to signify the Schr{\"o}dinger interpretation state): 
\begin{align*}
\rho^{S}(T_{n-1})&=U_{\tau_{n-1}}\rho(T_{n-2})U_{\tau_{n-1}}^*
\end{align*}
After the measurement has taken place (returning measured value $y$), but still at $T_{n-1}$, the state is given by:
\begin{align*}
j_{T_{n-2}}(\mathcal{P}_y)\rho^S(T_{n-1})j_{T_{n-2}}(\mathcal{P}_y)&=j_{T_{n-2}}(\mathcal{P}_y)U_{\tau_{n-1}}\rho(T_{n-2})U_{\tau_{n-1}}^*j_{T_{n-2}}(\mathcal{P}_y)
\end{align*}
The time evolution from $T_{n-1}$ to $T_n$, then gives:
\begin{align*}
&U_{\tau_n}j_{T_{n-2}}(\mathcal{P}_y)\rho^S(T_{n-1})j_{T_{n-2}}(\mathcal{P}_y)U_{\tau_n}^*\\
&=U_{\tau_n}j_{T_{n-2}}(\mathcal{P}_y)U_{\tau_{n-1}}\rho(T_{n-2})U_{\tau_{n-1}}^*j_{T_{n-2}}(\mathcal{P}_y)U_{\tau_n}^*
\end{align*}
Therefore, the joint probability for the measurements $x,y$ at times $T_{n}$ and $T_{n-1}$ respectively is given by:
\begin{align*}
p(x,y)&=Tr[U_{\tau_n}j_{T_{n-2}}(\mathcal{P}_y)U_{\tau_{n-1}}\rho(T_{n-2})U_{\tau_{n-1}}^*j_{T_{n-2}}(\mathcal{P}_y)U_{\tau_n}^*j_{T_{n-2}}(\mathcal{P}_x)]\\
&=Tr[j_{T_{n-2}}(\mathcal{P}_y)U_{\tau_{n-1}}\rho(T_{n-2})U_{\tau_{n-1}}^*j_{T_{n-2}}(\mathcal{P}_y)U_{\tau_n}^*j_{T_{n-2}}(\mathcal{P}_x)U_{\tau_n}]\\
&=Tr[U_{\tau_{n-1}}U_{\tau_{n-1}}^*j_{T_{n-2}}(\mathcal{P}_y)U_{\tau_{n-1}}\rho(T_{n-2})U_{\tau_{n-1}}^*j_{T_{n-2}}(\mathcal{P}_y)U_{\tau_n}^*j_{T_{n-2}}(\mathcal{P}_x)U_{\tau_n}]\\
&=Tr[j_{T_{n-1}}(\mathcal{P}_y)\rho(T_{n-2})U_{\tau_{n-1}}^*j_{T_{n-2}}(\mathcal{P}_y)U_{\tau_n}^*j_{T_{n-2}}(\mathcal{P}_x)U_{\tau_n}U_{\tau_{n-1}}]\\
&=Tr[j_{T_{n-1}}(\mathcal{P}_y)\rho(T_{n-2})U_{\tau_{n-1}}^*j_{T_{n-2}}(\mathcal{P}_y)U_{\tau_{n-1}}U_{\tau_{n-1}}^*U_{\tau_n}^*j_{T_{n-2}}(\mathcal{P}_x)U_{\tau_n}U_{\tau_{n-1}}]\\
&=Tr[j_{T_{n-1}}(\mathcal{P}_y)\rho(T_{n-2})j_{T_{n-1}}(\mathcal{P}_y)U_{\tau_{n-1}}^*U_{\tau_n}^*j_{T_{n-2}}(\mathcal{P}_x)U_{\tau_n}U_{\tau_{n-1}}]\\
&=Tr[j_{T_{n-1}}(\mathcal{P}_y)\rho(T_{n-2})j_{T_{n-1}}(\mathcal{P}_y)j_{T_n}(\mathcal{P}_x)]
\end{align*}
where the last line follows from the Markov property. The result for the conditional probability (as opposed to the joint probability) follows from the Bayes probability law, which serves to normalize the state at $T_{n-1}$.
\end{proof}
Therefore, we can see that at each time-step the impact of the measurement is applied to the state whereas the stochastic process is applied to the observable as before. In order to apply proposition \eqref{hybrid_prop} in a Monte-Carlo simulation, we use the conditional probability: $p(x|y)$ given by:
\begin{align*}
p(x|y)&=Tr[\rho_yj_{\tau}(\mathcal{P}_x)]\\
E^{\rho_y}[f(j_{\tau}(X)]&=Tr[\rho_yf(j_{\tau}(X))]
\end{align*}
Where $\tau$ is the Monte-Carlo time-step, and $\rho_y$ is defined by proposition \eqref{gen_meas_prop}.
\begin{proposition}{General Market Measurement:}\label{gen_meas_prop}
Let $j_t(X)$ be the stochastic process defined by equation \eqref{dX}, where the volatility is defined by equation \eqref{L_no_Ham}. Assume that the result of a measurement of the change in $j_t(X)$ from $t=0$ to $t=T$ is the value $x$. Then after the measurement, the state is represented by:
\begin{align}\label{eigenvector}
\rho_x&=|\psi_{mkt}\rangle\langle\psi_{mkt}|\otimes\sum_{k=1}^K|q_k|^2|s_k\rangle\langle s_k|\otimes|\varepsilon_k\rangle\langle\varepsilon_k|\\
|\varepsilon_k\rangle\langle\varepsilon_k|&=\frac{\mathcal{P}^{\Gamma}_k\Big(\int_0^TidA^{\dagger}_t-idA_t\Big)|\psi_{\Gamma}(0)\rangle\langle\psi_{\Gamma}(0)|\mathcal{P}^{\Gamma}_k}
{Tr\Big[\mathcal{P}^{\Gamma}_k\Big(\int_0^TidA^{\dagger}_t-idA_t\Big)|\psi_{\Gamma}(0)\rangle\langle\psi_{\Gamma}(0)|\mathcal{P}^{\Gamma}_k\Big]}\nonumber
\end{align}
where $\mathcal{P}^{\Gamma}_k$ projects onto the eigenspace for the Brownian motion to take the value:
\begin{align}\label{Brownian}
z_k=\frac{x}{\sigma_k}\pm \epsilon
\end{align}
and $\epsilon$ (small) reflects the precision of the measurement. The state $\rho_x$ will return the measurement $x$ (to the required precision) regardless of the choice of $q_k$.
\end{proposition}
\begin{proof}
The change in the price over time $T$ is given by:
\begin{align*}
\int_0^Tdj_t(X)&=1\otimes\sum_{k=1}^K\sigma_k|s_k\rangle\langle s_k|\otimes\Big(\int_0^TidA^{\dagger}_t-idA_t\Big)
\end{align*}
Now consider a projective measurement, and assume that we obtain the value $x$. Consider a vector $|\phi_x\rangle$:
\begin{align*}
|\phi_x\rangle&=|\psi_{mkt}\rangle\otimes\sum_{k=1}^Kq_k|e_k\rangle\otimes|\varepsilon_k\rangle
\end{align*}
For any values of $q_k$, such that $\sum_{k=1}^K|q_k|^2=1$, we have:
\begin{align*}
\bigg(\int_0^Tdj_t(X)\bigg)|\phi_x\rangle&=x|\phi_x\rangle
\end{align*}
\end{proof}
\begin{remark}
Since any combination of $q_k$ will yield an eigenstate for the measured price: $x$ (to the required precision), further modelling assumptions are required regarding what occurs as a result of the measurement, and specifically what information around the market volatility one obtains. We suggest two possibilities. First we consider the case whereby one obtains information regarding the market volatility as well as the market price. It should be noted that we have:
\begin{align*}
[j_t(\alpha),j_t(X)]&=0
\end{align*}
which means that measuring each of variables is possible, without impacting the other.

Secondly, we explore the case whereby one obtains no information regarding the volatility state. Whilst we have obtained the measured quantity $x$, referring to equation \eqref{Brownian}, we do not know whether we have a large volatility ($\sigma_k$) and small measured value for the Brownian motion ($z_k$), or a small volatility and large measured value for the Brownian motion.
\end{remark}
\begin{definition}[Joint Volatility and Price Measurement]\label{M1}
For the joint volatility and price measurement, we assume that after trading with the market we obtain information regarding both the volatility of the market, and its' current price. We represent this measurement process by $\mathcal{M}_1$, so that the collapsed state after the measurement is given by:
\begin{align}\label{M1_state}
\mathcal{M}_1(\rho)&=\frac{\mathcal{P}_k\rho\mathcal{P}_k}{Tr[\mathcal{P}_k\rho\mathcal{P}_k]}\\
\mathcal{P}_k&=1\otimes|s_k\rangle\langle s_k|\otimes\mathcal{P}^{\Gamma}_k\nonumber
\end{align}
So that after a measurement is taken, with the measured value $x$ together with volatility $\sigma_m$, the market state is given by:
\begin{align*}
\rho_x&=|\psi_{mkt}\rangle\langle\psi_{mkt}|\otimes|s_m\rangle\langle s_m|\otimes|\varepsilon_m\rangle\langle\varepsilon_m|\\
|\varepsilon_n\rangle\langle\varepsilon_n|&=\frac{\mathcal{P}^{\Gamma}_n\Big(\int_0^TidA^{\dagger}_t-idA_t\Big)|\psi_{\Gamma}(0)\rangle\langle\psi_{\Gamma}(0)|\mathcal{P}^{\Gamma}_n}
{Tr\Big[\mathcal{P}^{\Gamma}_n\Big(\int_0^TidA^{\dagger}_t-idA_t\Big)|\psi_{\Gamma}(0)\rangle\langle\psi_{\Gamma}(0)|\mathcal{P}^{\Gamma}_n\Big]}
\end{align*}
Where, as above, under $\mathcal{P}^{\Gamma}_n$, the Brownian motion takes the value:
\begin{align*}
z_n&=\frac{x}{\sigma_n}\pm\epsilon
\end{align*}
\end{definition}
\begin{definition}[Bayesian Approach]\label{M2}
Under the measurement process given by definition \eqref{M1}, by trading the market we obtain information regarding both the price and the volatility. If we assume, that the only result of the measurement is the price $x$, we can use the Bayesian rule to calculate the values for $q_k$. If we assume that the volatility state is given by:
\begin{align*}
\rho_{\sigma}&=\sum_{k=1}^KP(\sigma_k)|s_k\rangle\langle s_k|
\end{align*}
Then we can set $q_k$ using the probability of finding $\sigma_k$, given the measured value for the price, $x$:
\begin{align}\label{weights}
|q_k|^2&=P\big(\sigma_k|j_T(X)\in [x-\epsilon,x+\epsilon]\big)\\
&=\frac{P(j_T(X)\in [x-\epsilon,x+\epsilon]|\sigma_k])P(\sigma_k)}{\sum_{l=1}^LP(X_t\in [x-\epsilon,x+\epsilon]|\sigma_l)P(\sigma_l)}\nonumber
\end{align}
To calculate the weights, we have that $P(\sigma_k)$ are specified in the initial condition ($\rho_{\sigma}$, equation \eqref{prod_state}) and further that:
\begin{align*}
P(j_T(X)\in [x-\epsilon,x+\epsilon]|\sigma_k])&=\frac{1}{\sqrt{2\pi\sigma_k^2t}}\int_{x-\epsilon}^{x+\epsilon}exp\Big(-\frac{u^2}{2\sigma_k^2t}\Big)du
\end{align*}
\end{definition}
\section{A Comment on Unitary vs Non-Unitary Time Evolution:}\label{Comment_U_NU}
Before proceeding further, we clarify two different ways we can model the evolution of the operator: $j_t(X)$, and the resulting probability distribution for the results of possible future measurements. First, in section \eqref{MC} we consider the {\em unitary} evolution of $j_t(X)$, whereby we calculate the resulting probability distribution using a Monte-Carlo simulation. After each Monte-Carlo time-step, we make a measurement of the type described in definition \eqref{M1}.

Next, we consider the non-unitary evolution of  $j_t(X)$, where we calculate the resulting probability distribution using a partial differential equation approach.

Then in figure \eqref{UvsNU}, we show that this difference is not just numerical inaccuracy (for example Finite Difference discretization vs Monte-Carlo noise). The two different ways of considering future evolution of the price incorporate differing amounts of market information as we go along, and yield materially different results for the probability distribution.
\subsection{Unitary Evolution: Monte-Carlo Simulation}\label{MC}
We now consider running a Monte-Carlo simulation. Broadly speaking, we look at the following process:
\begin{enumerate}
\item[1)] The underlying evolves by a small time-step. For example, we look at the evolution over a time-step of 1 business day, which (depending on what calendar assumptions you apply) equates to roughly $dt\approx 0.004$.
\item[2)] This evolution is unitary, being driven by the unitary operator described by the stochastic process \eqref{unitary}. This stochastic process adds some noise to the Boson Fock space.
\item[3)] Once a measurement is taken, the degree of noise added depends on the value of the volatility eigenvalue: $\sigma_k$, that one gets from the measurement.
\item[4)] Since there is no system Hamiltonian, the system remains in this chosen volatility eigenstate for that path.
\item[5)] The system then starts again, at step 1), ready to undergo further unitary time evolution. Again, the system remains in the chosen volatility eigenstate, since there is no system Hamiltonian.
\end{enumerate}
In figure \eqref{Q_MC} we show the results of 1M Monte-Carlo paths, where we have used $K=31$ with the $\rho_{\sigma}$ being given by \eqref{rho_sig}. We have used $\sigma_k=0.04+(k*0.01)$, $i=1,\dots,31$. The simulated distribution is strongly non-Gaussian, with an excess kurtosis of 57\%.
\begin{figure}
\includegraphics[scale=0.9]{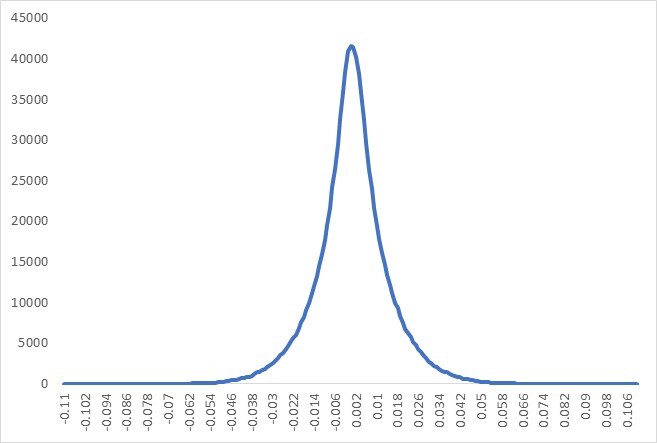}
\caption{Simulation results from 1M Monte-Carlo paths. There is 57\% excess Kurtosis.}\label{Q_MC}
\end{figure}
\subsection{Non-Unitary Time Evolution: PDE Approach}\label{non_u_app}
In order to consider the PDE method, as applied to the model defined by Hilbert space \eqref{Hilbert_2} and variance operator \eqref{L_no_Ham}, we first define an operator valued function:
\begin{align*}
F:\mathbb{R}^+\times \mathcal{L}[\mathcal{H}]\rightarrow\mathcal{L}[\mathcal{H}]
\end{align*}
where $\mathcal{L}[\mathcal{H}]$ represents linear operators on the Hilbert space $\mathcal{H}$, and here $\mathbb{R}^+$ represents the time axis. Then we can define an expectation at time $t$, conditional on the value at $t=0$:
\begin{align}\label{expectation}
u(t,x)&=E^{\rho}[f(t,j_t(X))|j_0(X)=x]\\
E^{\rho}[A]&=Tr[\rho A]\nonumber
\end{align}
We can derive a Kolmogorov backward equation by expanding in powers of $dj_t(X)$:
\begin{align*}
dF(t,j_t(X))&=\frac{\partial F}{\partial t}dt+\frac{\partial F}{\partial x}(t, j_t(X))dj_t(X)+\frac{1}{2}\frac{\partial^2 F}{\partial x^2}dj_t(X^2)
\end{align*}
then adjusting the expectation defined by \eqref{expectation}, to ensure $F(t,j_t(X))$ is a Martingale:
\begin{align*}
E^{\rho}[dF(t,j_t]&=0
\end{align*}
This leads to:
\begin{align}\label{KBE}
E^{\rho}\Big[\frac{\partial F}{\partial t}+\frac{1}{2}\frac{\partial^2 F}{\partial x^2}dj_t(X^2)\Big]&=0\nonumber\\
\frac{\partial u}{\partial t}+\frac{1}{2}\frac{\partial^2 u}{\partial x^2}\Big(\frac{1}{31}\sum_{k=1}^{31}\sigma_k^2\Big)&=0\nonumber\\
\frac{\partial u}{\partial t}+0.02\frac{\partial^2 u}{\partial x^2}&=0
\end{align}
Where in the second line we have applied the state given by \eqref{rho_sig}, with $K=31$, and $\sigma_k=0.04+(k*0.01)$, $i=1,\dots,31$, as per section \eqref{MC}. It is clear that equation \eqref{KBE} represents the Kolmogorov Backward equation for a Gaussian process with an annualized volatility of 20\%.
\begin{figure}
\includegraphics[scale=0.75]{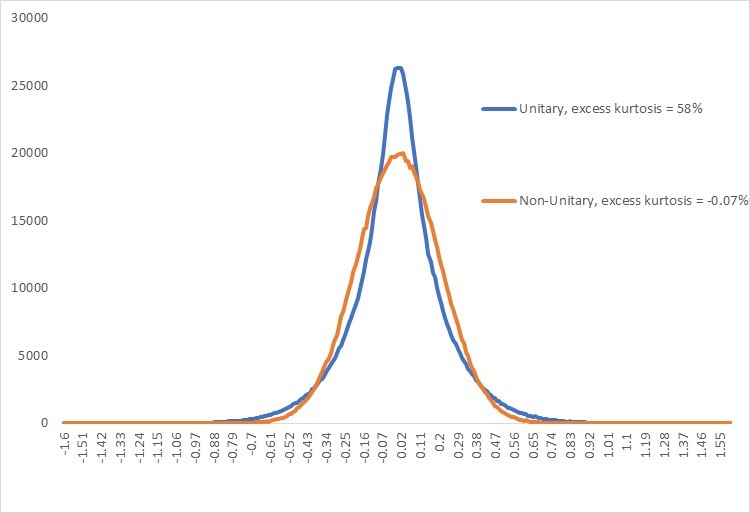}
\caption{Results from 1M Monte-Carlo paths, 1Y simulation. We contrast the probability distribution obtained from the PDE, to the results from the unitary time evolution. The volatility space is $\mathbb{C}^{31}$, with equally spaced eigen-values ranging from 5\% to 35\%. Under the non-unitary evolution the action of the partial trace means that we do not track the changing volatility state as time progresses.}\label{UvsNU}
\end{figure}

Figure \eqref{UvsNU} shows the results from a 1Y Monte-Carlo simulation of the probability distribution underlying the solution to equation \eqref{KBE}, versus the 1Y simulation for the unitary process described in section \eqref{MC}. To understand the difference, let:
\begin{align*}
\rho&=\rho_{mkt}\otimes\rho_{\sigma}\otimes\rho_{\Gamma}
\end{align*}
When we apply expectation \eqref{expectation}, what we want to do is average out the impact of the noise, and of the uncertain volatility level. Thus, we wish to look at the evolution of operators acting on a reduced density matrix:
\begin{align}\label{red_dens}
\rho_{mkt}&=Tr_{\mathcal{H}_{\sigma}\otimes\Gamma}[\rho]
\end{align}
Crucially, by taking the partial trace, we destroy the unitary nature of the time evolution. We lose the information contained in the volatility space, and in the Fock space $\Gamma$. The evolution of operators acting on a reduced density matrix such as \eqref{red_dens} is in general non-unitary. See \cite{BP} for more detail.

Note that if we do not take the partial trace over the volatility space, we will end up with $K$ different partial differential equations, and will only know which we need to solve, once the measurement has been taken.
\section{Modelling with Uncertain Volatility Case I:}\label{caseI}
\subsection{Volatility Hamiltonian Function:}
In the stochastic process defined by \eqref{unitary} the drift term is defined by the system Hamiltonian: $H$. By assuming $H=0$, the observable $j_t(X)$ is a Martingale in the sense that, taking expectation over the Fock space, to yield an operator acting on $\mathcal{H}_{mkt}\otimes\mathcal{H}_{\sigma}$, we have:
\begin{align*}
E^{\Gamma}[j_T(X)|j_t(X)]&=j_t(X)
\end{align*}
In this section we consider the impact of setting:
\begin{align}\label{Hamiltonian}
H&=\mathbb{I}\otimes H_{\sigma}\text{, }H_{\sigma}=\frac{\nu^2}{2}\Delta_{\sigma} + V_{\sigma}
\end{align}
Here, $\Delta_{\sigma}$ represents a kinetic energy term allowing the volatility state to evolve. $V_{\sigma}$ represents the potential energy.
\begin{proposition}\label{vol_prop_1}
Let the system Hilbert space be given by \eqref{Hilbert_2}, and let the time evolution be described by the unitary operator: $U_t$ in equation \eqref{unitary}, with $L$ given by \eqref{L_no_Ham} and $H$ by \eqref{Hamiltonian}. Then, under unitary time evolution, the traded price operator follows the following process:
\begin{align}
dj_t(X)&=j_t(\alpha)dA_t+j_t(\alpha^{\dagger})dA^{\dagger}_t\text{, }dj_t(\alpha)=j_t(\theta_{\alpha})dt\nonumber\\
\alpha&=[L^*,X]\text{, }\theta_{\alpha}=i[H,\alpha]\nonumber
\end{align}
\end{proposition}
\begin{proof}
First note that using the unitary process given by \eqref{unitary}, we have:
\begin{align}\label{dX,dalpha}
dj_t(X)&=j_t(\theta_x)dt+j_t(\alpha)dA_t+j_t(\alpha^{\dagger})dA^{\dagger}_t\\
dj_t(\alpha)&=j_t(\theta_{\alpha})dt+j_t(\gamma)dA_t+j_t(\gamma^{\dagger})dA^{\dagger}_t\nonumber\\
\alpha&=[L^*,X]\text{, }\gamma=[L^*,\alpha]\nonumber\\
\theta_x&=i[H,X]-\frac{1}{2}[L^*LX+XL^*L-2L^*XL]\nonumber\\
\theta_{\alpha}&=i[H,\alpha]-\frac{1}{2}[L^*L\alpha+\alpha L^*L-2L^*\alpha L]\nonumber
\end{align}
Under \eqref{Hamiltonian} we find that: $[H,X]=0$, so that:
\begin{align*}
dj_t(X)&=j_t(\alpha)dA_t+j_t(\alpha^{\dagger})dA^{\dagger}_t
\end{align*}
Furthermore, we have: $[L^*,\alpha]=0$, so that $\alpha$ evolves deterministically under the Hamiltonian: $H=\mathbb{I}\otimes\Big(\frac{\nu^2}{2}\Delta_{\sigma}+V_{\sigma}\Big)$
\end{proof}
\subsection{A Monte-Carlo Simulation:}
In order to setup a simulation for the process with non-zero Hamiltonian: \eqref{Hamiltonian}, we first clarify the projective measurement process that we wish to use (see for example \cite{NC} section 2.2). From above, we have:
\begin{align*}
\alpha&=[L^*,X]\\
&=-i\otimes\sum_{k=1}^K\sigma_k|s_k\rangle\langle s_k|\otimes\mathbb{I}
\end{align*}
Each Monte-Carlo step in the simulation described in section \eqref{MC} is $dt=0.004\approx 1$ day. In the non-zero kinetic energy case we proceed as follows:
\begin{enumerate}
\item[1)] At time $t_n=ndt$, we take a measurement of the type described in definition \eqref{M1}. This results in $\sigma_{k_n}$. The resulting collapsed state is given by equation \eqref{M1_state}. In other words, the result of the measurement is a collapse in the volatility space to yield a volatility eigenvalue, and in the Fock space, yielding the measured value for the Brownian noise.
\item[2)] Apply the unitary operator: $U_{dt}=exp(-iH_{\sigma}dt)$ to reflect the unitary evolution of the volatility state over time-step $dt$:
\begin{align*}
\rho(t_{n+dt})&=U_{dt}\mathcal{M}_1(\rho(t_n))U_{dt}^{\dagger}
\end{align*}
\item[3)] Now the probability of obtaining $\sigma_{k_2}$ (conditional on having measured $\sigma_{k_1}$ at the beginning of the time-step) is given by:
\begin{align}\label{vol_cond_prob}
p(\sigma_{k_2}|\sigma_{k_1})&=Tr\Big[\mathcal{M}_1\Big(U_{dt}\mathcal{M}_1\big(\rho(t_n)\big)U_{dt}^{\dagger}\Big)\Big]
\end{align}
\item[4)] So, again assuming the prior time-step volatility was $\sigma_{k_1}$, we must make a random selection using equation \eqref{vol_cond_prob}, to select the volatility for the next time-step.
\end{enumerate}
\subsection{Zero Potential Energy Case:}
In this section, let us assume that the volatility space is given by the infinite dimensional space: $\mathcal{H}_{\sigma}=L^2(\mathbb{\mathcal{K}})$, where $\mathcal{K}$ is a compact subset of $\mathbb{R}$. Since we have redefined the volatility Hilbert space, we must also now redefine the operator $L$, acting on $\mathcal{H}_{mkt}\otimes\mathcal{H}_{\sigma}\otimes\Gamma$:
\begin{align}\label{L_inf}
L&=-i\frac{\partial}{\partial x}\otimes\int_{\mathcal{K}}\sigma d\mu^L(\sigma)\otimes\mathbb{I}
\end{align}
Where $\mu^L$ is a projection valued measure (as given for example in \cite{Hall} Theorem 7.12). Therefore, we find:
\begin{align}\label{com_inf}
[L^*,X]&=-i\otimes\int_{\mathcal{K}}\sigma d\mu^L(\sigma)\otimes\mathbb{I}
\end{align}
We can now set the Hamiltonian \eqref{Hamiltonian} to that defining the free Sch{\"o}dinger equation:
\begin{align}\label{free_Ham}
H_{\sigma}&=-\frac{\nu^2}{2}\frac{\partial^2}{\partial\sigma^2}
\end{align}
Using this Hamiltonian, we find that (see for example \cite{Hall} Theorem 4.5) \eqref{vol_cond_prob} is given by (where $\delta t$ is the time-step):
\begin{align}\label{prob_free}
p(\sigma_{k_2}|\sigma_{k_1})&=\sqrt{\Big(\frac{1}{2\pi\nu^2\delta t}\Big)}\exp\Big(-\frac{(\sigma_{k_2}-\sigma_{k_1})^2}{2\nu^2\delta t}\Big)
\end{align}
\subsection{High vs Low Energy Case:}\label{hi_vs_lo}
Using proposition \eqref{vol_prop_1}, with the conditional probability that results from the measurement process given by equation \eqref{prob_free}, we can configure different levels of volatility uncertainty. To illustrate, we consider the unitary evolution described by updating steps 1) to 6) from section \eqref{MC}.
\begin{enumerate}
\item[1)] We simulate the random evolution of the price operator $j_t(X)$ according to proposition \eqref{vol_prop_1}.
\item[2)] We assume we follow unitary time evolution. That is, after each time-step we take a measurement, and the system collapses into a volatility eigenstate.
\item[3)] As before, the stochastic process adds noise into the Boson Fock space.
\item[4)] Again, as before, the degree of noise added depends on the value of the volatility that one gets from the measurement process.
\item[5)] However, now there is a Hamiltonian whereby at the time of the next measurement, the volatility is no longer in a fixed eigenstate: $|s_k\rangle\langle s_k|$ associated with the measured volatility eigenvalue: $\sigma_k$.
\item[6)] Therefore, when we come to make the next random time-step, and make another measurement of the system, we apply equation \eqref{prob_free} to calculate the probability for the volatility conditional on the previously observed value.
\item[7)] Thus at each time-step we require {\em two} random numbers. One to simulate the random evolution of the price variable, and one to select the volatility.
\end{enumerate}
\begin{enumerate}
\item[{\em Max Energy case:}] If one assumes $\nu\rightarrow\infty$, then the volatility will quickly drift back to a state whereby each eigenvalue is equally likely. In effect, the system ``forgets'' the measured value for the volatility very quickly. The result is that at each successive time step, each volatility is equally likely. Therefore, along each path the time-step variance eventually averages out to:
\begin{align*}
dt\int_{\mathcal{K}}\sigma^2d\mu^L(\sigma)&=\sigma_0^2dt
\end{align*}
\item[{\em Zero Energy case:}] If one assumes $\nu\rightarrow 0$, then as soon as the first measurement is taken, the volatility is fixed along that path. When we look at the localized variance along each path, there is no averaging effect. If we set (per section \eqref{HJB_app}):
\begin{align*}
\mathcal{K}&=[\sigma_{min},\sigma_{max}]
\end{align*}
Then some paths will have a constant time-step variance close to $\sigma_{min}^2dt$, and some a constant time-step variance close to $\sigma_{max}^2dt$. Therefore, the uncertainty in the volatility is maximised, and we have a {\em non-Gaussian} evolution, with maximal excess kurtosis. Ie, this represents the case shown in figure \eqref{Q_MC}.
\end{enumerate}
\subsection{Simulation Results:}
We run the simulation as per section \eqref{MC}. However, rather than fixing the volatility, we select a volatility by randomly sampling from the probability distribution \eqref{prob_free} at each step. Figure \eqref{MC_res1} shows the resulting simulated distributions for a simulation with $\sigma_{min}=0.01\%$, $\sigma_{max}=15\%$, and for $\nu=1\%$ up to $\nu=326\%$, and 100K paths. Similarly, figure \eqref{MC_res2} shows the implied volatility surface resulting from the simulation. 
\begin{figure}
\includegraphics[scale=0.45]{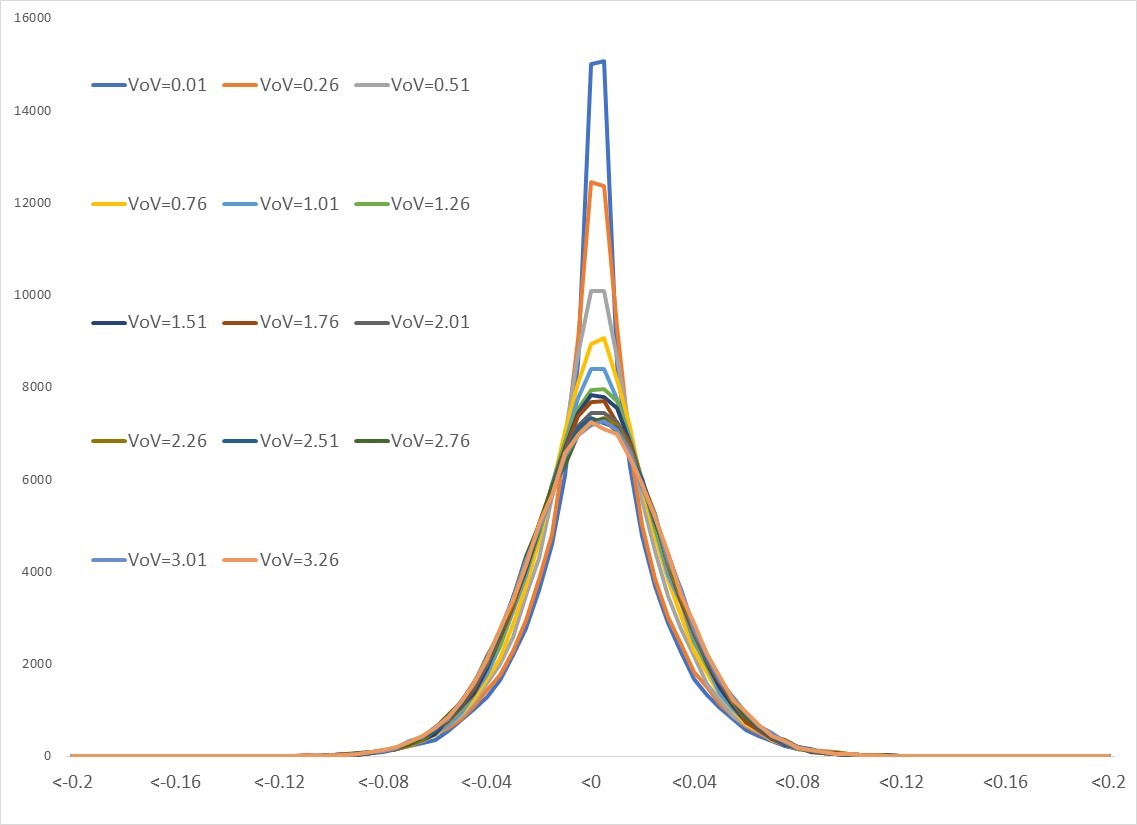}
\caption{Probability distribution results from 1 month Monte-Carlo simulation, 100K paths. Horizontal axis shows $X$, vertical axis shows number of paths.}\label{MC_res1}
\end{figure}
\begin{figure}
\includegraphics[scale=0.6]{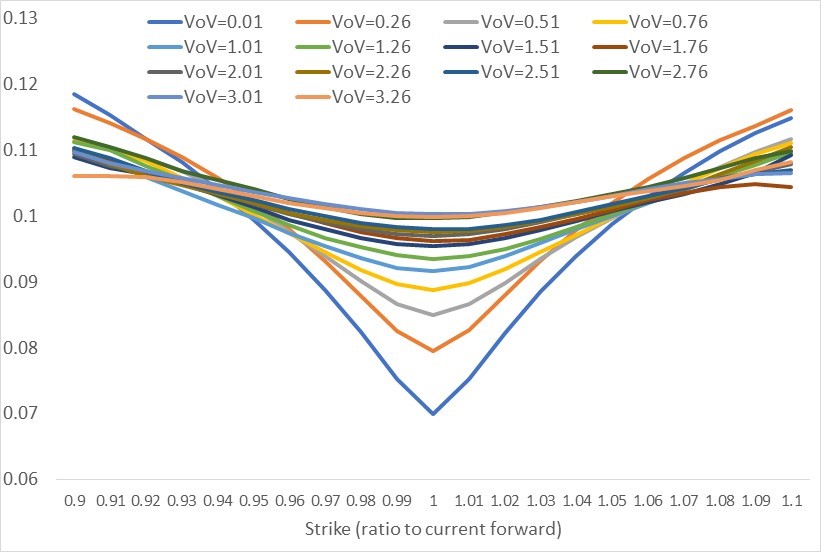}
\caption{Simulated implied volatility surfaces from a 1 month Monte-Carlo simulation, 100K paths. Horizontal axis shows the strike, vertical axis shows the implied volatility.}\label{MC_res2}
\end{figure}
As described in section \eqref{hi_vs_lo}, as $\nu^2\delta t$ increases, the Hamiltonian rapidly causes the level of uncertainty in the volatility to increase. Thus, as one moves along the path the volatility chosen becomes increasingly random. As the path gets longer, and more time-steps are added, this leads to an averaging out effect, and the resulting simulated distribution starts to look Gaussian in nature.

For lower values of $\nu^2\delta t$, the volatility is more likely to remain closer to the measured value, and the distribution resembles the fat tailed distributions shown in figures \eqref{Q_MC} and \eqref{UvsNU}.
\section{Modelling with Uncertain Volatility Case II:}\label{caseII}
In this section, we look again at the case of zero Hamiltonian, but using the Bayesian measurement approach described in definition \eqref{M2}, rather than the definition given in definition \eqref{M1}.
\subsection{Monte-Carlo Sampling vs Reality:}
In a standard Monte-Carlo simulation, for example under a local volatility model:
\begin{align*}
\delta x_n&=\sigma(x_{n-1},t_{n-1})\cdot z_n\\
z_n&\sim\mathcal{N}(0,\sqrt{\delta t_n})
\end{align*}
there is a sense that the Monte-Carlo path reflects reality. At each step along each path we track the price changes by successively drawing Gaussian random variables. The state of the market is defined by the price, which is what we are simulating.

In the case described by definition \eqref{M2}, the state of the market is defined by the volatility state, in addition to the price. To carry out the simulation we must select a random variable from the distribution:
\begin{align}\label{approx_dx2}
\delta x_n&\sim z_n\sum_{i=1}^K|q_i|^2\sigma_i\\
z_n&\sim\mathcal{N}(0,\sqrt{\delta t_n})
\end{align}
To do this, we first draw a Gaussian random number: $z_n$, before randomly selecting the volatility $\sigma_l$ based on the probability $|q_l|^2$ given by equation \eqref{weights}. Whilst this approach to Monte-Carlo simulation may converge to the true probability distribution for $\delta x_n$, it no longer reflects reality in the sense that we do not record the true volatility quantum state along each path.
\subsection{Monte-Carlo Simulation:}
We run a Monte-Carlo simulation based on equation \eqref{approx_dx2}, with 20 time-steps, each $\delta t=0.004\approx 1$day, and 100K simulation paths. We set $\mathcal{H}_{\sigma}=\mathbb{C}^K$, and investigate different values of $K$, and different values for $\epsilon$ in equation \eqref{weights}. We proceed using the following steps:
\begin{enumerate}
\item[1)] We assume that the volatility state starts in a maximum entropy state. Therefore, the initial weights are given by: $|q_k^0|^2=1/K$.
\item[2)] We select a Gaussian random variable from $\mathcal{N}(0,1)=z_k$, and randomly select $\sigma_k$ using the weights.
\item[3)] After each time-step we recalculate the weight for the next time-step: $|q^{n+1}_k|^2$ using equation \eqref{weights}. So for example if the previous time-step yielded the value $\delta x_n$, with randomly selected volatility $\sigma_k$, we would have:
\begin{align*}
|q^{n+1}_k|^2&=\frac{\frac{|q^n_k|^2}{\sqrt{2\pi\sigma_k^2\delta t}}\int^{\delta x_n+\epsilon}_{\delta x_n-\epsilon}\exp\Big(\frac{-u^2}{2\sigma_k^2\delta t}\Big)du}{\sum_{j=1}^K\frac{|q^n_j|^2}{\sqrt{2\pi\sigma_j^2\delta t}}\int^{\delta x_n+\epsilon}_{\delta x_n-\epsilon}\exp\Big(\frac{-u^2}{2\sigma_j^2\delta t}\Big)du}
\end{align*}
\item[4)] Once the weights for the next time-step are calculated, go back to step 2) and repeat the process.
\end{enumerate}
\subsection{Simulation Results:}
In this section we show the results from a Monte-Carlo simulation with 20 time-steps of $\delta t=0.004$, and 100K paths. We set the volatility space to: $\mathbb{C}^{31}$, with volatility ranging from 5\% to 35\% in 1\% intervals.

Figure \eqref{bayes_chart} shows the resulting probability distribution vs the equivalent simulation for the joint measurement approach. This shows that the difference is small. The joint measurement approach results in 58\% excess kurtosis, vs 56\% for the Bayesian approach. To investigate the reason for this small difference, we look at the individual paths taken for the simulated volatility, as subsequent measurements are taken.
\begin{figure}
\includegraphics[scale=0.35]{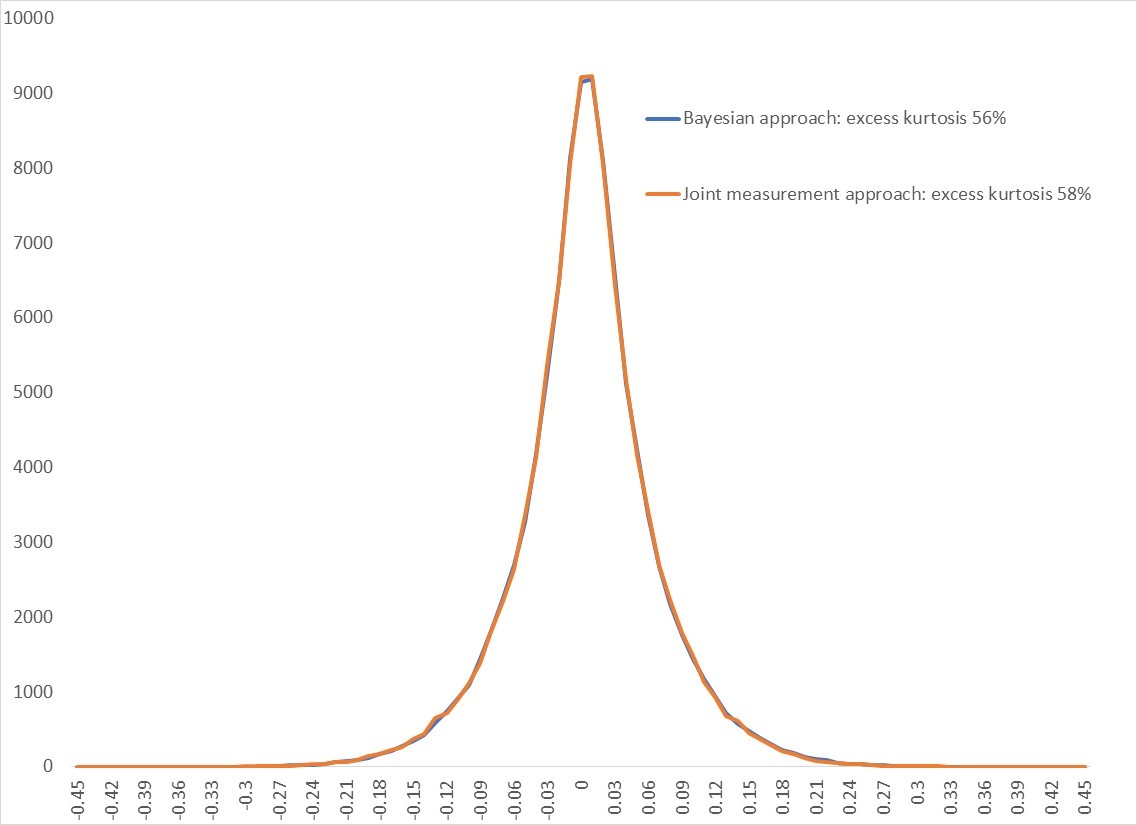}
\caption{Probability distribution for 100K paths, showing that the difference between the joint measurement approach, and the Bayesian approach, is small.}\label{bayes_chart}
\end{figure}
\begin{figure}
\includegraphics[scale=0.35]{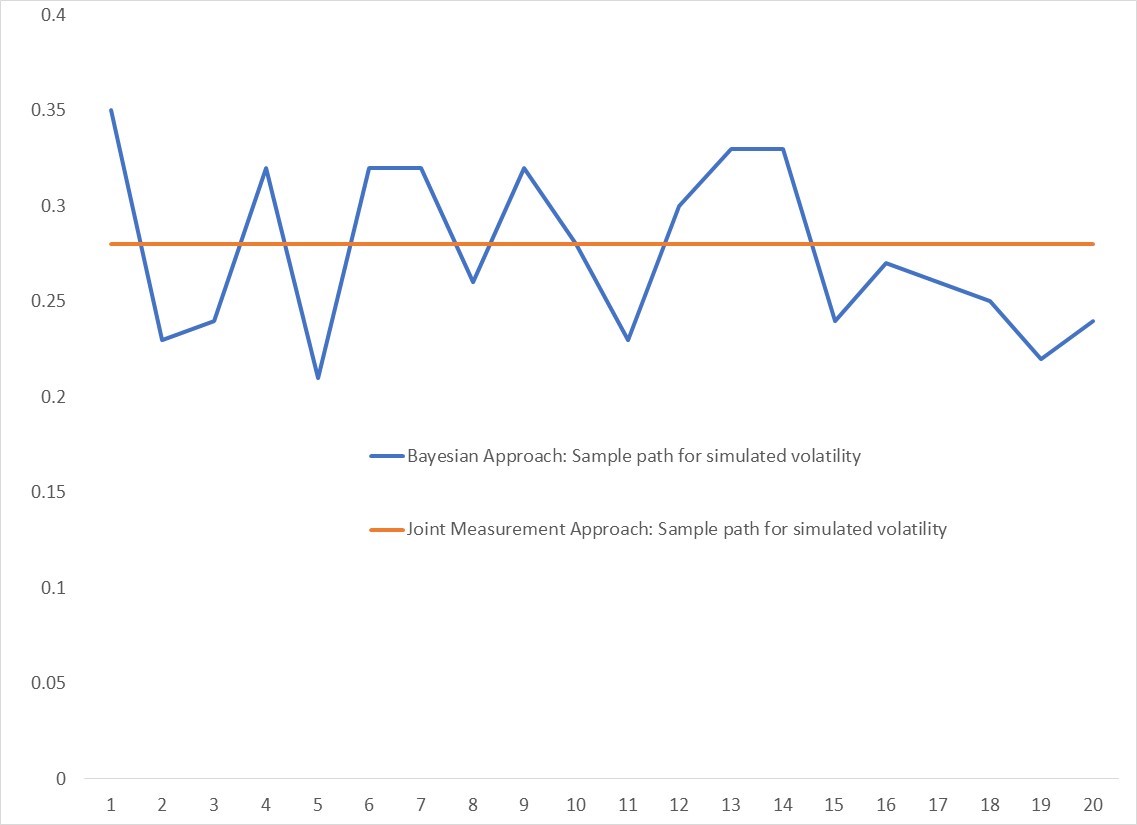}
\caption{Example path for simulated volatility, Bayes measurement vs joint measurement approaches.}\label{bayes_volpath}
\end{figure}
Figure \eqref{bayes_volpath} shows a sample path for the simulated volatility in each case.

Under the joint measurement approach, we fix on a volatility as a result of the measurement. This then does not change, due to the fact that there is no Hamiltonian. 

Under the Bayesian measurement approach, there is no fixed volatility, since this is not measured. However, as described above we randomly sample the volatility in such a way as to ensure the resulting probability distribution that results from the Monte-Carlo simulation converges to the theoretical distribution.

In order to assess the likelihood of a path with a high expected volatility from a path with a low expected volatility, figure \eqref{cond_evol} shows the expected volatility, conditional on the volatility at the previous time-step being greater than 20\%, and less than 20\%.

The separation between those paths that have a higher expected volatility and those with a lower expected volatility (based on the collapsed state at that step), for the Bayesian approach, quickly converges to a similar value to the joint measurement approach, where the volatility is fixed from the beginning.
\begin{figure}
\includegraphics[scale=0.35]{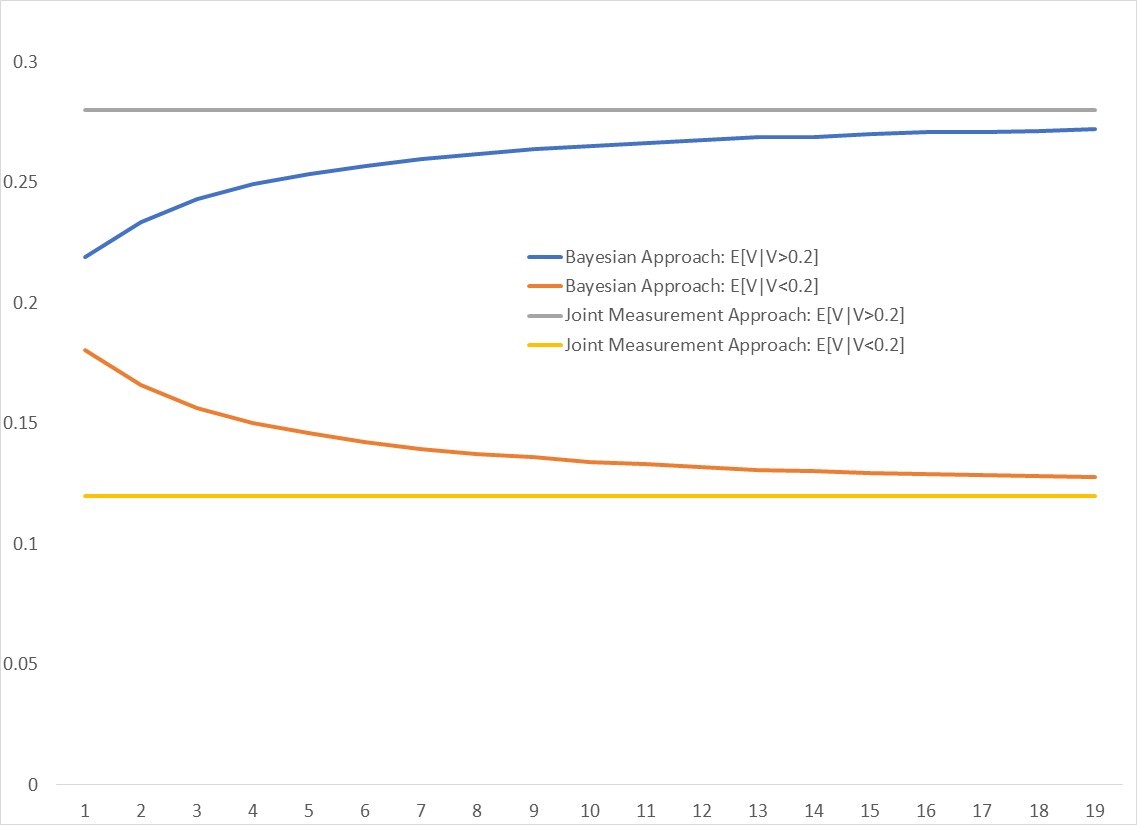}
\caption{The chart shows the expected volatility conditional on the volatility being greater than, and less than 0.2 vs the time-step.}\label{cond_evol}
\end{figure}
\section{Conclusion:}
In this article we have shown how to extend existing approaches to the modelling of the financial markets using quantum probability, by incorporating an additional Hilbert space for uncertain parameters.

In a conventional model, we fix the value of parameters in advance, meaning that the level of random noise coming from the stochastic process completely determines the changes in the variable one is measuring. We have shown in this article that once you allow for uncertainty in the parameters that configure the stochastic process, you open up different possibilities for the dynamics.
The evolution of the probability density function will depend on how much information regarding the state of the uncertain parameters is obtained during a measurement of the observables.

For example, one can construct a partial differential equation by considering the evolution of a reduced density matrix. This averages out the information contained in both the noise space, and the uncertain parameter space using a partial trace, and leads to a Gaussian probability distribution.

Alternatively one can allow for different approaches to measurement that retain the unitary nature of the evolution, and leads to non-Gaussian probability distributions.

%
%
\end{document}